\documentclass{siamltex}

\usepackage{acronym}		
\usepackage{amssymb}
\usepackage{amsmath}
\usepackage{epsfig}		
\usepackage{graphicx}
\usepackage{tensor}		
\usepackage{multirow}           
\usepackage[amssymb]{SIunits}	

\newtheorem{assumption}{Assumption}
\newtheorem{remark}{Remark}

\newcommand{\NA}{\mathbb{N}}         
\newcommand{\RE}{\mathbb{R}}         
\newcommand{\MG}{\mathbb{G}}        
\newcommand{\ra}{\mathrm{a}}         %
\newcommand{\rb}{\mathrm{b}}         %
\newcommand{\rd}{\mathrm{d}}         
\newcommand{\rt}{\mathrm{t}}
\newcommand{\cL}{\mathcal{L}}	     
\newcommand{\cM}{\mathcal{M}}	     
\newcommand{\cO}{\mathcal{O}}        
\newcommand{\qH}{{\hat{H}}}
\newcommand{\qp}{{\hat{p}}}
\newcommand{\tm}{{\tilde{\mu}}}
\newcommand{\tn}{{\tilde{\nu}}}

\DeclareMathOperator{\id}{Id}

\acrodef{ODE}{ordinary differential equation}
\acrodef{DAE}{differential-algebraic equation}
\acrodef{PH}{port-Hamiltonian}
\acrodef{IDA}{interconnection and damping assignment}
\acrodef{PBC}{passivity-based control}
\acrodef{BCH}{Baker-Campbell-Hausdorff}
\acrodef{KYP}{Kalman-Yakubovich-Popov}

\title{Discrete-Time Models for \\ Implicit Port-Hamiltonian Systems}

\author{
 Fernando Casta\~nos\footnotemark[1] \and  
 Hannah Michalska\footnotemark[2] \and
 Dmitry Gromov\footnotemark[3] \and
 Vincent Hayward\footnotemark[4]
 }

\begin{document}

\maketitle

\renewcommand{\thefootnote}{\fnsymbol{footnote}}
\footnotetext[1]{Departamento de Control Automatico, Cinvestav del IPN, Av. IPN 2508, Col. San Pedro Zacatenco,
 C.P. 07360, M\'exico D.F., M\'exico (\tt{castanos@ieee.org})}
\footnotetext[2]{McGill Centre for Intelligent Machines, 3480 University Street, Montreal Quebec, Canada 
 H3A 2A7 (\tt{hannah.michalska@mcgill.ca})}
\footnotetext[3]{Faculty of Applied Mathematics, St.Petersburg State University, University pr., 35, Petrodvorets, 198 904 St.Petersburg, Russia (\tt{dv.gromov@gmail.com})}
\footnotetext[4]{UPMC Paris 6, UMR7222, Institut de Syst\`emes et de Robotique, Paris, France}
\renewcommand{\thefootnote}{\arabic{footnote}}

\begin{abstract}
Implicit  representations of finite-dimensional port-Hamiltonian systems are studied from the perspective of their use in numerical simulation and control design. Implicit representations arise when a system is modeled
in Cartesian coordinates and when the system constraints are applied in the form of additional algebraic equations (the system model is in a~\acl{DAE} form). Such representations lend themselves better to sample-data approximations.
Given an implicit representation of a port-Hamiltonian system it is shown how to construct a sampled-data model that preserves the \ac{PH} structure under sample and hold.
\end{abstract}

\begin{keywords} 
 port-Hamiltonian systems, nonlinear implicit systems, symplectic integration, sampled-data systems
\end{keywords}

\pagestyle{myheadings}
\thispagestyle{plain}
\markboth{Fernando Casta\~nos, Hannah Michalska, Dmitry Gromov and Vincent Hayward}{Discrete-Time Models for Implicit port-Hamiltonian Systems}

\section{Introduction}

\acresetall

The class of Hamiltonian systems has a prominent role in many disciplines. It  was recently extended in~\cite{dalsmo1999}
to include open systems, i.e. systems that interact with the environment via a set of inputs and outputs (called \emph{ports}),
giving rise to \ac{PH} systems. Such extended models immediately reveal the passive properties of the underlying systems, making
them particularly well suited for designing \ac{PBC} laws. Two types of model representations of Hamiltonian systems are in widespread
use: the \emph{explicit representation} stated in the form of an \ac{ODE} on an abstract
manifold~\cite{rheinboldt1984,rheinboldt1991,reich1990,reich1991} and the \emph{implicit representation} stated in the form of a
\ac{DAE} usually evolving in a Euclidean space~\cite{castanos2013}.  While explicit representations are usually preferred in the
context of analytical mechanics; see~\cite{arnold2,goldstein,marsden}, the implicit \acp{DAE} models lend themselves better for numerical
computations as they lead to simpler expressions for the Hamiltonian functions. The formal relations between the two representations and
their equivalence can be established if the system's configuration space is regarded as an embedded submanifold of the Euclidean space;
see~\cite{castanos2013} for a full development.

Of principal interest here will be the construction of sampled-data (discrete-time) models of \ac{PH} systems that best approximate their
continuous-time counterparts. Sampled-data models are important in digital control implementations and permit for simpler design of \ac{PBC} 
laws directly in discrete time. In this context, the notion of ``best approximation'' deserves clarification.

For linear systems, \emph{exact} sampled-data models can be constructed by requiring that the solutions of the sampled-data and continuous-time
systems coincide at the discretization points. Point-wise model matching is usually impossible in the case of nonlinear systems short of explicit
derivation of analytical expressions for their solutions. For general dynamic systems, it is thus the choice of an integration method for
generating an approximate solution (like Euler or Runge-Kutta) whose precision relates to the compromise between complexity and order of
approximation of the given continuous-time system by its sampled-data counterpart. In the special case of Hamiltonian systems, structure
preservation is usually the main criterion for choosing a numerical method; see, e.g.,~\cite{hairer} for more details on structure-preserving numerical
schemes. Structure-preserving methods guarantee accuracy in long-horizon simulations.

It is known that autonomous Hamiltonian systems conserve two quantities: the Hamiltonian function $H$ (i.e., the energy or storage function) and a
certain two-form $\omega$, called the symplectic form. Numerical integration algorithms can either conserve $H$~\cite{laila2006,laila2007} or
$\omega$~\cite{reich1996}, but not both.
Conservation of $\omega$ is usually preferred over conservation of $H$ as the symplectic form unambiguously defines the class of Hamiltonian vector
fields; see Theorem~\ref{thm:poincare} and Remark~\ref{rem:sympl}.

\subsection*{Contribution}
Adopting a symplectic-form preserving approach, a sampled-data model for a given continuous-time \ac{PH} system in developed here.
The approach employs implicit representations of \ac{PH} systems as the latter lead to simpler expressions for the Hamiltonian functions.
Specifically, the Hamiltonian functions arising from implicit representations lend themselves well to the application of flow-splitting
numerical integration methods; see~\cite{reich1996} for a splitting method that applies to autonomous Hamiltonian systems (Hamiltonian
systems without ports). Additionally, the discrete-time modeling approach presented here preserves passivity of \ac{PH} systems
(Theorem~\ref{thm:main}). The passive structure is preserved in the sense that the discrete model is the \emph{exact} representation of
another, possibly non affine, continuous-time \ac{PH} system which, up to an approximation error of order two with respect to the sampling
interval, has the same storage function $H$ and the same output function $y$.

\subsection*{Paper structure}
Section~\ref{sec:pH} introduces the implicit representations of \ac{PH} systems. This section recalls the results from~\cite{castanos2013}
and serves mainly to introduce the notation and to state the main assumptions (\ref{ass:rankg} and~\ref{ass:elliptic}). Also, the definition
of \ac{PH} systems is extended to the case of non affine control systems. Splitting and symplectic integration methods for autonomous systems
are recalled in~\S~\ref{sec:auto}. Discrete-time models for implicit \ac{PH} systems based on vector splitting methods are developed
in~\S~\ref{sec:main} and their asymptotic properties are analyzed. An illustrative example is provided confirming the findings.

\section{Implicit representations of \acl{PH} systems} \label{sec:pH}

We begin this section by defining the configuration and phase spaces of Hamiltonian systems in implicit form. Conservative properties
are also discussed. Similarly to the case of explicit representations of Hamiltonian systems, both the Hamiltonian function and the symplectic
two-form of the implicit representation are invariant under the flow of the system. Following~\cite{dalsmo1999}, input and output ports are
added to the implicit representation, insuring passivity of the resulting \ac{PH} system. Finally, we extend the definition of \ac{PH} to the
the case of non affine control systems. This will be needed when performing backward error analysis later in~\S~\ref{sec:main}
since time discretization entails the loss of affinity.

\subsection{The configuration space}

The class of systems considered here is restricted to mechanical finite-dimensional systems with
scleronomic constraints that evolve in continuous time. Systems of this type typically consist of
$M$ rigid bodies held together as one structure by the action of constraint forces. The position of
each of the rigid bodies can be unambiguously described in terms of the position of its center of mass
and its orientation in space. The configuration space of a spatial system as a whole can thus be viewed as
a subset, $\MG$, of a Cartesian space of dimension $n = 6M$ that is expressed in terms of smooth, independent
constraint functions $g : \RE^n \to \RE^k$, with $k \le n$, as the level set
\begin{equation} \label{eq:ambient}
 \MG := g^{-1}(0) = \left\{ r \in \RE^n \;|\; g(r) = 0 \right\}  
\end{equation}
where $g^{-1}(0)$ is the inverse image of $0 \in \RE^k$.
Functional independence of the constraints is expressed in terms of the following \emph{rank condition}.

\begin{assumption} \label{ass:rankg} 
The rank of $g$ is equal to $k$ at all points of the set $g^{-1}(0)$. 
\end{assumption}

Recalling that the rank of a mapping is the rank of its tangent map (pushforward
by $g$), then in any coordinates, the condition requires that the Jacobian 
$G := \left\{ \partial g_i/\partial r_j \right\}_{ij}$
has full rank, i.e., $\rank(G) = k$ at every point of $g^{-1}(0)$.
The value $0 \in \RE^k$ is then called a \emph{regular value} of $g$, the level set $g^{-1}(0)$
is \emph{a regular level set} of $g$, and $g$ is said to be a \emph{defining map} for $g^{-1}(0)$;
see~\cite{lee}, p. 113-114. It hence follows that $\MG$ can be given the differentiable structure
of a closed embedded submanifold of $\RE^n$; see~\cite{lee}, Corollary 5.24, p. 114. The Constant-Rank
Level Set Theorem (\cite{lee}, Theorem 5.22,  p. 113), further specifies its dimension as $o=n-k$.

Because $\MG$ is an embedded submanifold, it can also be regarded as an abstract $o$-dimensional 
manifold with local coordinates $q$ (called \emph{generalized coordinates} in mechanics). It is then also
possible to exhibit an injective inclusion map: $\imath : \MG \hookrightarrow \RE^n$ (embedding), with
$\rank(\imath) = \dim \MG$, which is a homeomorphism onto the image $\imath(\MG) \subset \RE^n$ and such
that $\imath(q)=r$ and $g \circ \imath \equiv 0$. This inclusion serves as a map between local
coordinates $q$ and global Cartesian coordinates $r$.

\subsection{Hamiltonian equations} \label{sec:HamEq}

Let $T^*\MG$ be the cotangent bundle of $\MG$ and let $\left\{ q^i,\qp_i \right\}$ be local coordinates. A system is
said to be \emph{Hamiltonian} if its trajectories are integral curves of the Hamiltonian vector field
$D_\qH : T^*\MG \to T(T^*\MG)$,
\begin{displaymath}
 D_{\qH} = \frac{\partial \qH}{\partial \qp_i}\frac{\partial}{\partial q^i} - 
  \frac{\partial \qH}{\partial q^i}\frac{\partial}{\partial \qp_i} \;,
\end{displaymath}
where $\qH : T^*\MG \to \RE$ is the sum of the system's kinetic and potential energies $\hat{K} : T^*\MG \to \RE$ and
$\hat{V} : \MG \to \RE$, respectively; see~\cite{arnold2} for more details and a coordinate-free definition.

The implicit model for a Hamiltonian system is defined as follows. Let $T^*\RE^n$ be the cotangent bundle of
$\RE^n$ and let $\left\{ r^i,p_i \right\}$ be global coordinates. The implicit Hamiltonian vector field takes the form
\begin{equation} \label{eq:XHg}
 X_{H,g} = D_H - \lambda_j \frac{\partial g^j}{\partial r^i} \frac{\partial }{\partial p_i} \;, \quad g = 0 \;,
\end{equation}
with
\begin{equation} \label{eq:DH}
 D_H = \frac{\partial H}{\partial p_i}\frac{\partial}{\partial r^i} - \frac{\partial H}{\partial r^i}\frac{\partial }{\partial p_i}
\end{equation}
the unconstrained part of the Hamiltonian vector field and $H : T^*\RE^n \to \RE$ is again the sum of the kinetic and potential
energies, but expressed in global coordinates. That is, $K : T^*\RE^n \to \RE$ and $V : \RE^n \to \RE$;
see~\cite{hairer,reich1996,arnold3} for details on the derivation of this equation. The vector field unravels as the DAE system
\begin{subequations} \label{eq:HamC}
\begin{align}
 \dot{r} &= \nabla_p H(r,p) \;, \quad \dot{p} = -\nabla_r H(r,p) - G(r)^\top\lambda \\
       0 &= g(r) \label{eq:HamCon} \;.
\end{align}
\end{subequations}

By applying $X_{H,g}$ to both sides of the constraint equations $g^j = 0$, one obtains the so-called \emph{hidden constraints} 
\begin{equation} \label{eq:f}
 f^j := X_{H,g}(g^j) = \frac{\partial H}{\partial p_i}\frac{\partial g^j}{\partial r^i} = 0 \;.
\end{equation}
Thus, the system evolves on the closed submanifold
\begin{equation} \label{eq:LG}
 L\MG = \left\{ (r,p) \in T^*\RE^n \;|\; g(r) = 0 , f(r,p) = 0 \right\}
\end{equation}
and we have $X_{H,g} : L\MG \to T(L\MG)$. A more rigorous construction of the phase space $L\MG$ is given in~\cite{castanos2013},
where it is regarded as an embedding of $T^*\MG$ in $T^*\RE^n$. The map $\qp \mapsto p$ as well as the formal relation between
$\qH$ and $H$ can also be found in~\cite{castanos2013}.

The Lagrange multipliers $\lambda_j$ are defined implicitly by~\eqref{eq:XHg}. Precisely, application of 
$X_{H,g}$ to the hidden constraints makes $\lambda_j$ appear:
\begin{equation} \label{eq:f2}
 X_{H,g}(f^l) = D_H(f^l) - \lambda_j \frac{\partial g^j}{\partial r^i} \frac{\partial f^l}{\partial p_i} = 0 \;.
\end{equation}
Thus, if the matrix
\begin{equation} \label{eq:f3}
 \left\{ \frac{\partial g^j}{\partial r^i} \frac{\partial f^l}{\partial p_i} \right\}_{jl} = 
  \left\{ \frac{\partial g^j}{\partial r^i} \frac{\partial^2 H}{\partial p_i \partial p_m} \frac{\partial g^l}{\partial r^m} \right\}_{jl}
\end{equation}
is non-singular on $L \MG$, then there are unique $\lambda_j$ satisfying~\eqref{eq:f2} and forcing the integral curves to stay on $L \MG$.

\begin{assumption} \label{ass:elliptic}
The Hessian matrix
$\left\{ \partial^2 H(r,p)/\partial p_i \partial p_j \right\}_{ij}$
is positive definite for all $(r,p) \in T^*\RE^n$ so $H(r,p)$ is convex in $p$.
\end{assumption}

This assumption, satisfied by most mechanical systems, together with Assumption~\ref{ass:rankg} ensures that the matrix in~\eqref{eq:f3}
is invertible and the system is well defined. In mechanical systems, $\lambda$ is the covector of constraint forces that insure satisfaction of
the constraints during the evolution of the system.

\subsection{Energy conservation and symplecticity}

It is well known that the flows generated by Hamiltonian vector fields in explicit form preserve the Hamiltonian
function, i.e. the total energy of the system is conserved during the evolution of the
system~\cite{schaft,ortega2001,ortega2002}. As the explicit and implicit system representations are equivalent,
the conservation also holds for the implicit Hamiltonian vector field~\eqref{eq:XHg}, as is easily confirmed by
computing the Lie derivative of $H$ along the flow generated by $X_{H,g}$~\cite{hairer}. Indeed,
\begin{equation} \label{eq:LXH}
 \cL_{X_{H,g}}H = X_{H,g}(H) = D_H(H) - \lambda_j \frac{\partial g^j}{\partial r^i} \frac{\partial H}{\partial p_i} = \lambda_j f^j \;,
\end{equation}
where the third equality holds because $D_H(H) = 0$
and because of the definition of $f$ given in~\eqref{eq:f}. Eq.~\eqref{eq:LXH} shows that 
$\cL_{X_{H,g}}H \big|_{L \MG} = 0$
(cf. Eq.~\eqref{eq:LG}), so $H$ remains constant along the system trajectories.

It is also well known that, besides the 0-form $H$, Hamiltonian flows preserve a certain 2-form called the
symplectic form. In the case of implicit representations the symplectic form is given by the
formula\footnote{See~\cite{arnold2} for a coordinate-free definition.}
\begin{equation} \label{eq:om}
 \omega := \rd r^i \wedge \rd p_i \;,
\end{equation}
which acts on vectors of $T(T^*\RE^n)$, with Einstein's summation convention implied.

\begin{definition}
 A differentiable mapping $\phi : T^* \RE^n \to T^*\RE^n$ is called \emph{symplectic} if
 \begin{equation} \label{eq:sympl} 
  \phi^*\omega = \omega \;.
 \end{equation}
\end{definition}
Here, $ \phi^*$ denotes the pull-back map associated with $\phi$, defined by
\begin{definition} \label{def:pullb}
Let $\phi : \cM_1 \rightarrow \cM_2$ be a smooth map of manifolds and let $p \in \cM_1$. The pull-back map 
$\phi^* : T^*_{\phi(p)} \cM_2 \rightarrow T^*_p \cM_2$ associated with $\phi$
is a dual map to the push-forward map $\phi_*$ and is characterized by
\begin{displaymath}
 \langle \phi^* \xi , X \rangle = \langle \xi , \phi_* X \rangle \ \text{ for } \ \xi \in T^*_{\phi(p)} \cM_2 , \ X \in T_p \cM_1 \;.
\end{displaymath}
\end{definition}
The application of this definition permits to re-write (\ref{eq:sympl}) in the equivalent form
\begin{equation} \label{eq:sympl2}
 \omega (\phi_* \xi,\phi_* \eta) = \omega (\xi,\eta) \quad \text{for all } \xi,\eta \in T(T^*\RE^n) \;.
\end{equation}

The conservation of $\omega$ along $X_{H,g}$ can be established by showing that the Lie derivative $\cL_{X_{H,g}}\omega$
is equal to zero, the demonstration bearing similarity to that of the conservation of $H$; see~\cite{castanos2013}
for the explicit derivation. Here, we cite the result

\begin{theorem}\cite{hairer,leimkuhler} \label{thm:poincare}
 Let $H$ be twice continuously differentiable. The flow $\phi_t : L \MG \to L \MG$ of $X_{H,g}$ governed by~\eqref{eq:XHg} is a symplectic transformation on $L \MG$, i.e.,
 \begin{displaymath}
  \phi_t^*\omega = \omega
 \end{displaymath}
 for every $t$ for which $\phi_t$ is defined.
\end{theorem}

\begin{remark} \label{rem:sympl}
 The converse statement, that \emph{every symplectic flow $\phi_t$ solves Hamilton's equations for some $H$}, is also true, so
 symplecticity is a characteristic property of Hamiltonian systems~\cite{arnold2}. This does not translate to the case of energy
 conservation, i.e., while every Hamiltonian system conserves energy, not every energy-conserving system is Hamiltonian.
\end{remark}

\subsection{port-Hamiltonian systems}

In the presence of external forces and dissipation it is convenient to represent~\eqref{eq:XHg} as an
input-output system equipped with a pair of port variables $(u,y)$, giving rise to a \emph{\ac{PH} system};
see~\cite{maschke1992,schaft,ortega2001} for the original definition as stated with respect to Hamiltonian
systems in explicit form. Extending on this definition the Hamiltonian systems in implicit form, a
port-Hamiltonian system is described in terms of the vector field $X_{H,u,g} : L \MG \times(\RE^m)^* \to T(L \MG)$:
\begin{equation} \label{eq:XHug}
 X_{H,u,g} = D_H + \left(u_lU\indices{_i^l} - \lambda_j \frac{\partial g^j}{\partial r^i} \right)\frac{\partial }{\partial p_i} \;,
  \quad g = 0 \;,
\end{equation}
with  $u \in (\RE^m)^*$ defined as the controlled or input variable, $y \in \RE^m$ is defined as the output variable that satisfies
\begin{equation} \label{eq:ry}
 y^l = U\indices{_i^l}\frac{\partial H}{\partial p_i}
\end{equation}
and where $U\indices{_i^l}$ are maps from $\RE^n$ to $\RE$.

The vector field~\eqref{eq:XHug} and the output unravel as the DAE system
\begin{align*}
 \dot{r} &= \nabla_p H(r,p) \;, \quad \dot{p} = -\nabla_r H(r,p) - G(r)^\top\lambda + U(r)u \\
       y &= U(r)^\top \nabla_p H(r,p) \\  
       0 &= g(r) \;.
\end{align*}

By analogy with the results described in Sec.~\ref{sec:HamEq}, one can determine the Lagrange multipliers $\lambda$ explicitly. The
constraints $f^a = 0$ imply that
\begin{equation} \label{eq:fu}
 X_{H,u,g}(f^a) = D_H(f^a) + u_lU\indices{_i^l}\frac{\partial f^a}{\partial p_i} -
  \lambda_j \frac{\partial g^j}{\partial r^i} \frac{\partial f^a}{\partial p_i} = 0 \;,
\end{equation}
from where it follows that, as long as~\eqref{eq:f3} is non-singular, there are unique $\lambda_j$ (in general dependent
on $u$ as well as on $r$ and $p$) such that $X_{H,u,g}(f^a) = 0$ and such that the integral curve stays on $L \MG$. 

It can be readily seen that an implicit \ac{PH} system described by~\eqref{eq:XHug} no longer preserves $H$. The Lie derivative of $H$ is now
\begin{displaymath}
 \cL_{X_{H,u,g}}(H) = X_{H,u,g}(H) = D_H(H) + u_lU\indices{_i^l}\frac{\partial H}{\partial p_i} -
  \lambda_j \frac{\partial g^j}{\partial r^i} \frac{\partial H}{\partial p_i}  = u_ly^l - \lambda_j f^j \;,
\end{displaymath}
which establishes the power balance
\begin{equation} \label{eq:pb}
 \cL_{X_{H,u,g}}H \big|_{L \MG} = u_ly^l \;.
\end{equation}
Since the product $u_ly^l$ is equal to the rate of change in energy, we say that $(u,y)$ is a power-conjugated pair of port variables. If,
in addition, the restriction of $H$ to $L \MG$ is bounded from below, i.e., if the image of $L \MG$ under $H$ is bounded from below, then 
~\eqref{eq:XHug} is called passive, or more precisely, \emph{lossless}. Boundedness of $H$ can be easily assessed using the following
proposition.

\begin{proposition}\cite{castanos2013}
 If the potential energy $V$ is lower semi-continuous and $\MG$ is compact, then $H$ restricted to $L \MG$ is bounded from below (hence, the vector
 field~\eqref{eq:XHug} describes a lossless system).
\end{proposition}


With the inclusion of the control variable $u$, it can no longer be expected that the flow
of~\eqref{eq:XHug} be symplectic. Indeed, it is not hard to see that the Lie derivative of
$\omega$ along $X_{H,u,g}$ is in general different from zero.

\begin{proposition}\cite{castanos2013}
 The Lie derivative of $\omega$~\eqref{eq:om} along~\eqref{eq:XHug} and restricted to $L \MG$
 satisfies
 \begin{equation} \label{eq:LieOm}
  \cL_{X_{H,u,g}} \omega \big|_{L \MG}  = \rd r^i \wedge \rd(u_lU\indices{_i^l}) \;.
 \end{equation}
\end{proposition}

\subsubsection*{Example: A double planar pendulum}

Let us recall an example from~\cite{castanos2013}, on which we will elaborate when discussing sampled-data models. 

Consider the model of a double planar pendulum shown in Fig.~\ref{fig:doubleP} that comprises a pair of point masses $m_a$ 
and $m_b$ whose coordinate positions are $r^a = (r^{a_x},r^{a_y})$ and $r^b = (r^{b_x},r^{b_y})$, respectively.
The massless bars are of fixed lengths $l_a$ and $l_b$ which gives rise to the two holonomic constraints:
\begin{equation} \label{eq:const}
 g^1(r) = \| r^a \|^2 - l_a^2 = 0  \;, \quad
 g^2(r) = \| r^\delta \|^2 - l_b^2  = 0 
\end{equation}
with $r := (r^a,r^b) \in \RE^n$,  $n = 4$, $k=2$, $r^\delta := r^b - r^a$. The rank of the constraint Jacobian is full as
\begin{equation} \label{eq:rank}
 \rank G(r) = \rank
  \begin{pmatrix}
   r^{a_x}      &  r^{a_y}      & 0 &          0 \\
  -r^{\delta_x} & -r^{\delta_y} & r^{\delta_x} & r^{\delta_y}
  \end{pmatrix}
  = k 
\end{equation}
for all $r \in \MG$. Therefore, $0$ is a regular value of $g$ and $\MG$ is an embedded submanifold of $\RE^4$.


The total energy is
\begin{equation} \label{eq:HEx}
 H(r,p) = \frac{1}{2} p^\top M^{-1} p + \bar{g}(m_a r^{a_y} + m_a r^{a_y}) \;,
\end{equation}
where
$M := \begin{pmatrix} m_a \mathbf{I}_n & \mathbf{0}_n \\ \mathbf{0}_n & m_b \mathbf{I}_n \end{pmatrix}$
and $\bar{g}$ is the standard gravity.

Substituting~\eqref{eq:const} and~\eqref{eq:HEx} in~\eqref{eq:HamC} gives
\begin{subequations} \label{eq:doublePC}
\begin{align}
 \dot{r}^a &= m_a^{-1} p_a \;, \quad \dot{r}^b = m_b^{-1} p_b \\
 \begin{pmatrix}
  \dot{p}_{a_x} \\ \dot{p}_{a_y} \\ \dot{p}_{b_x} \\ \dot{p}_{b_y} 
 \end{pmatrix}
  &= -
 \begin{pmatrix}
  0 \\ \bar{g}m_a \\ 0 \\ \bar{g}m_b
 \end{pmatrix}
  - 2
 \begin{pmatrix}
  r^{a_x} & -r^{\delta_x} \\
  r^{a_y} & -r^{\delta_y} \\
        0 & r^{\delta_x} \\
        0 & r^{\delta_y}
 \end{pmatrix}
 \begin{pmatrix}
  \lambda_1 \\ \lambda_2
 \end{pmatrix}
\end{align}
\end{subequations}
which, together with~\eqref{eq:const}, constitutes a set of \acp{DAE} describing the motion of the double pendulum in implicit form. The multipliers
$\lambda_1$ and $\lambda_2$ are the magnitudes of the tension along the two bars. 

\begin{figure}
\begin{center}
 \includegraphics[width=0.35\textwidth]{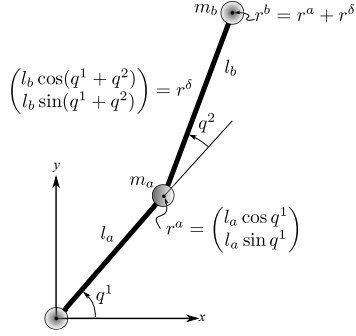}
\end{center}
\caption{A double planar pendulum, a simple Hamiltonian system.}
\label{fig:doubleP}
\end{figure}

Now, assume that the double pendulum 
is actuated by application of torques $u_1$ and $u_2$ to the joints that
correspond to the angles $q^1$ and $q^2$, respectively. The resulting linear forces are then $U^1 u_1$ and $U^2 u_2$ with
\begin{equation} \label{eq:U}
 U^1 := \{ U\indices{_i^1} \}_i = 
  \begin{pmatrix}
   -r^{a_y} \\ r^{a_x} \\ 0 \\ 0
  \end{pmatrix} 
  \frac{1}{l_a^2} \quad \text{and} \quad
 U^2 := \{ U\indices{_i^2} \}_i = 
 \begin{pmatrix}
   r^{\delta_y} \\ -r^{\delta_x} \\ -r^{\delta_y} \\ r^{\delta_x} 
  \end{pmatrix} 
  \frac{1}{l_b^2} - U^1 \;.
\end{equation}
The manifold defined by~\eqref{eq:const} is compact and the potential energy is continuous, which confirms that the double
pendulum is passive with passive outputs
$y^l = U\indices{_i^l} \frac{\partial H}{\partial p_i} = U\indices{_i^l}\dot{r}^i$.

An explicit model for the double pendulum as an \ac{ODE} can be derived by choosing the generalized coordinates on 
$\MG$ as $q^1 \in (-\pi,\pi)$ and $q^2 \in (-\pi,\pi)$, as motivated by Fig.~\ref{fig:doubleP}. The associated
embedding $r = \imath(q)$ satisfying $g \circ \imath \equiv 0$  is readily exhibited as 
\begin{equation} \label{eq:chart}
 \begin{pmatrix}
  r^{a_x} \\ r^{a_y} \\ r^{b_x} \\ r^{b_y} 
 \end{pmatrix} 
 = 
  \begin{pmatrix}
   l_a \cos q^1 \\ l_a \sin q^1 \\ l_a \cos q^1 + l_b \cos q^\rt \\ l_a \sin q^1 + l_b \sin q^\rt
  \end{pmatrix} \;, \quad q^\rt := q^1 + q^2 \;.
\end{equation}
In local coordinates, the total energy is
\begin{equation} \label{eq:HhEx}
 \qH(q,\qp) = \frac{1}{2}\qp^\top \hat{M}(q)^{-1}\qp + \bar{g}\left( m_\rt l_a\sin q^1 + m_b l_b \sin q^\rt \right) \;.
\end{equation}
with 
\begin{displaymath}
 \hat{M}(q) = \begin{pmatrix} m_\rt l_a^2 + m_b l_b^2 + 2 m_b l_a l_b \cos q^2 \!&\! m_b l_b^2 + m_b l_a l_b \cos q^2  \\ m_b l_b^2 + m_b l_a l_b \cos q^2 & m_bl_b^2 \end{pmatrix}
\end{displaymath}
The motion of the system is described by
\begin{equation}\label{eq:doublePM}
 \dot{q} =  \hat{M}(q)^{-1} \qp \;, \; \dot{\qp} = -\nabla_q V(q) - \nabla_q \left( \frac{1}{2} \qp^\top \hat{M}(q)^{-1} \qp \right) + u
\end{equation}
(see~\cite{castanos2013} for more details).

Two representations for the same system were derived, one in the form of an \ac{ODE}~\eqref{eq:doublePM} and the other as a \ac{DAE}~\eqref{eq:doublePC}.
The main point of this example can be summarized in the following remark.

\begin{remark} \label{rem:sepa}
The implicit \ac{DAE} representation of the Hamiltonian systems renders a separable Hamiltonian function~\eqref{eq:HEx}, i.e. the kinetic energy does not
depend on $r$ and the inertia matrix $M$ is a \emph{constant} diagonal matrix. Moreover, the potential energy is \emph{linear},
which results in a constant potential energy gradient. On the other hand it is easily verified that the explicit Hamiltonian representation does not share
such fortunate properties. The explicit Hamiltonian is no longer separable and the inertia matrix depends on the generalized coordinates $q$.
\end{remark}
 
The a cost of a simpler Hamiltonian function is the higher dimensional implicit model representation and the appearance of the Lagrange multipliers.
As will soon be seen, however, implicit representations are particularly advantageous for the purposes of system discretization.

\subsection{Non affine \acl{PH} systems}

It was assumed in~\eqref{eq:XHug} that the control variables enter the vector field affinely. Many physical systems exhibit
this property so, from a modeling point of view, this assumption is not too restrictive. However, for the purpose of
performing backward error analysis in \S~\ref{sec:main}, we will need to consider \ac{PH} systems for which the control
might enter in a non affine way. Motivated by the fact that $\cL_{X_{H,u,g}} \omega \big|_{L \MG} = 0$ when $u \equiv 0$,
we propose the following extended definition of a \ac{PH} system.

\begin{definition} \label{def:PH}
 A controlled vector field $X : L \MG \times (\RE^m)^* \to T(L \MG)$ is said to be \emph{\acl{PH}} if $u \equiv 0$ implies
 that the generated flow is symplectic.
\end{definition}

Passivity of a non affine \acl{PH} system can be established by redefining the passive output.

\begin{lemma} \label{lem:PH}
 A (not necessarily affine) smooth \ac{PH} system described by a vector field $X$ can always be decomposed as
 $X = X_0 + u_lZ^l$, where $X_0 : L \MG \to T(L \MG)$ is a Hamiltonian vector field and $Z^l : L \MG \times (\RE^m)^* \to T(L \MG)$
 are the input vector fields. Hence, $X$ satisfies the power balance
 $X(H) = u_ly^l$
 for some real-valued function $H$ and real-valued output functions $y^l = Z^l(H)$. (The output functions may now depend directly on $u$
 as well as on $r$ and $p$.)
\end{lemma}

\begin{proof}
 Following~\cite{lin1995c}, we first show that a smooth control vector field can be split into a drift and a set of vector fields having
 $u$ factored out. Let us define the drift $X_0$ as $X_0(x) = X(x,0)$ and let us define the vector fields $W^l$ by the equations
 \begin{displaymath}
  W^l(\alpha) = \frac{\partial }{\partial u_l} \Big( X(\alpha) \Big) \quad \text{for all } \alpha \in \mathcal{C}^\infty(L \MG,\RE) \;. 
 \end{displaymath}
 It follows from the chain rule that
 \begin{displaymath}
  u_lW^l(x,\theta u)(\alpha) = \frac{\rd}{\rd \theta}\Big( X(x,\theta u) \Big)(\alpha) \;.
 \end{displaymath}
 Upon integration on both sides of the equation we arrive at
 \begin{displaymath}
  u_l\int_0^1 W^l(x,\theta u)(\alpha)\rd \theta = X(x,\theta u)(\alpha)\Big|_0^1 = X(x,u)(\alpha) - X_0(x)(\alpha) \;. 
 \end{displaymath}
 Therefore, we have
 \begin{equation} \label{eq:0Z}
  X = X_0 + u_lZ^l \;,
 \end{equation}
 where the input vector fields $Z^l$ are defined by
 \begin{displaymath}
  Z^l(x,u)(\alpha) = \int_0^1 W^l(x,\theta u)(\alpha)\rd \theta \quad \text{for all } \alpha \in \mathcal{C}^\infty(L \MG,\RE) \;.
 \end{displaymath}

 It follows from the hypothesis that $X_0$ generates a symplectic flow, so it is a Hamiltonian vector field and satisfies $X_0(H) = 0$ for
 some real-valued function $H$. Applying $X$ to $H$ shows that
 $X(H) = X_0(H) + u_lZ^l(H) = u_ly^l$
 with $y^l = Z^l(H)$.
\end{proof}

\begin{remark}
For an affine \ac{PH} system, the formulae of this lemma recover the output functions~\eqref{eq:ry} with
$Z^l = U\indices{_j^l}\frac{\partial}{\partial p_j}$, that is,
$y^l = Z^l(H) = U\indices{_i^l}\frac{\partial H}{\partial p_i}$.
\end{remark}

\section{Sampled-data models for autonomous Hamiltonian systems} \label{sec:auto}

Computing a sampled-data model of a dynamical system basically amounts to computing an approximate
solution of the differential equations during a small interval of time. This problem has been studied
extensively in the literature of numerical analysis, from which we borrow some results and terminology.
In numerical analysis, a \emph{sampled-data model} is the central component of an \emph{integration method}
or a \emph{numerical integrator}, up to the point that these terms are used interchangeably.

Mathematical models for sampled-data systems arise in diverse circumstances. In the direct approach to digital
control, i.e., as opposed to the emulation of continuous control laws, the design of the controller is performed
in discrete time, the designer working directly over a sampled-data model. When designing directly in discrete time,
the controller can be directly implemented on a digital device. Also, it is possible to exploit the advantages of
switched controls or, e.g., multirate control techniques~\cite{monaco2007b}.

Computing the sampled-data model for a given nonlinear system relies on the computation of a solution $\phi_t$ of 
the corresponding \ac{ODE} or \ac{DAE}, which is in general impossible to do analytically, so one has to settle for
an approximate solution.

When simulating the behavior of dynamic systems, a discrete-time model of the continuous system is also used for computing a numerical solution
to the initial-value problem. Many different integration methods (or \emph{methods} for short) can be found in the literature. 
Let us first focus on integration methods for autonomous systems and further restrict our attention to one-step
methods defined by a transformation
\begin{displaymath}
 \psi_h : x_\alpha \mapsto x_{\alpha+1} \;,
\end{displaymath}
where the constant step-size $h$ is regarded as a parameter of the method\footnote{For a more general method, the value of the $x_{\alpha+1}$
need not depend only on $x_{\alpha}$, but may also depend on the previous values $x_{\alpha-1}, x_{\alpha - 2}$,\dots (a multistep method).
Also, the value of $h$ need not be constant in general.}. For a given initial condition $x_0$ in the phase space, $\psi_h$ is applied
recursively to generate a discrete flow $x_1, x_2, x_3, \dots$ that approximates the true flow
$\phi_h(x_0), \phi_{2h}(x_0)$, $\phi_{3h}(x_0),\dots$ of a given vector field $X$ at time instants $h,2h,3h,\dots$ In this sense, the map
$\psi_h$ is a discrete-time approximation of $\phi_h$ (or a sampled-data model of $\phi_h$).

\begin{definition}
 A one-step method has \emph{order} $s$ if the local error satisfies\footnote{We use big-O notation when quantifying approximation errors,
 i.e., for a given pair of functions $e_1(h)$, $e_2(h)$, we write $e_1(h) = \cO(e_2(h))$ as $h \to a$ as shorthand for
 $\limsup_{h \to a} \frac{\|e_1(h)\|}{\|e_2(h)\|} < \infty$.
 }
 \begin{equation} \label{eq:forward}
  \psi_h(x_0) - \phi_h(x_0) = \cO(h^{s+1}) \quad \text{as} \quad h \to 0
 \end{equation} 
 uniformly in $x_0$. A one-step method is said to be consistent if $s \ge 1$.
\end{definition}

Let us now discuss some important properties of numerical integrators, like order and symmetry.

\subsection{Symplectic methods} \label{sec:sympl-methods}

If a sampled-data model approximates the discrete time behavior of a Hamiltonian system, one could hope for $\psi_h$ to inherit its
fundamental qualitative properties: energy conservation and symplecticity. Unfortunately, it is not possible to preserve $H$ and $\omega$
simultaneously, unless $\psi_h$ agrees with the exact flow $\phi_h$ up to a reparametrization of time~\cite{ge1988}. For this reason, one has
to choose either in favor of one or the other invariant\footnote{For particular Hamiltonians there might be other invariants, such as momentum
or angular momentum, but in general there need not be.}. Energy conserving methods have received some
attention~\cite{gonzalez,simo1991,gonzalez1999,laBudde1976a,laBudde1976b}, but in light of Remark~\ref{rem:sympl}, most of the literature
focuses on symplectic integration algorithms (see~\cite{hairer,leimkuhler,channell1996} and references therein). A comparison between both
approaches is carried out in~\cite{simo1992b} for the rigid body.

A theoretical advantage of constructing a symplectic one-step method is that, even though $\psi_h$ only approximates $\phi_h$ up to
the $s$'th order, it coincides \emph{exactly} (if one disregards convergence issues) with the flow of another Hamiltonian system, a
modified Hamiltonian system described by a \emph{modified differential equation}.

\begin{theorem}\cite[p. 352]{hairer} \label{th:modHam}
 A symplectic method $\psi_h : L \MG \to L \MG$ for the constrained Hamiltonian system~\eqref{eq:HamC} has a modified equation that is locally
 of the form
 \begin{subequations} \label{eq:modHamC}
  \begin{align}
   \dot{r} &= +\nabla_p \tilde{H}(r,p) \;, \quad \dot{p} = -\nabla_r \tilde{H}(r,p) - G(r)^\top \tilde{\lambda} \\
         0 &= g(r) 
  \end{align}
 \end{subequations}
 with $\tilde{H} = H + h H_2 + h^2H_{3} + \dots$ Furthermore,
 \begin{displaymath}
  \frac{\partial H_j(r,p)}{\partial p_i}\frac{\partial g^l(r)}{\partial r^i} = 0 \;, \quad \text{for all } (r,p) \in L \MG \;,
 \end{displaymath}
 all $l = 1,\dots,k$ and all $j$. Note that the actual value of the Legendre multipliers $\tilde\lambda$ differ in general from those obtained for
 the original system (\ref{eq:HamC}).
\end{theorem}

In other words, for an initial condition $x_0$, $\psi_h(x_0)$ is equal to the solution of~\eqref{eq:modHamC} at time $t=h$. Note that~\eqref{eq:forward} provides information about the difference between the actual flow $\phi_h$ of $X$ and the
approximate discrete flow $\psi_h$. This is the kind of information that \emph{forward} error analysis aims at. While certainly useful as an
indicator of the quality of the approximation, Eq.~\eqref{eq:forward} only evaluates the behavior of the approximate flow on the first
iteration, but says nothing about its long time behavior. From~\eqref{eq:forward} alone we cannot infer anything about the
error $x_\alpha - \phi_{\alpha h}(x_0)$ when $\alpha$ is large, so we do not know if errors accumulate or if they average out to zero.

On the other hand, Theorem~\ref{th:modHam} tells us that if $\psi_h$ is symplectic, then there exists a modified continuous system
whose flow coincides exactly with the discrete flow generated by $\psi_{h}$. The modified system~\eqref{eq:modHamC} preserves the Hamiltonian
structure of the original system~\eqref{eq:HamC} and it is `close' to it in the sense that $\tilde{H} = H + \cO(h^s)$ for a method of order
$s$. In other words, a symplectic integration method preserves the original 2-form $\omega$ and a different (but close) Hamiltonian function.
This property guaranties that the good behavior of the integration scheme is maintained during many iterations, giving a global nature to the
local property~\eqref{eq:forward}. This observation is at the center of \emph{backward error analysis}~\cite{hairer}.

\subsection{Splitting methods} \label{sec:split} 

A practical advantage of symplectic schemes is that they lend themselves well to the application of
splitting methods. To illustrate the idea, consider again the \emph{unconstrained} or, otherwise,
explicit Hamiltonian vector field $D_\qH$ on an abstract manifold $T^*\MG$. If the Hamiltonian function is separable, i.e., if
it can be written as $\qH(q,\qp) = \qH_\ra(q) + \qH_\rb(\qp)$, then the vector field can be spilt into two
Hamiltonian vector fields
\begin{displaymath}
 D_{\qH_\ra} = -\frac{\partial \qH_\ra}{\partial q^i}\frac{\partial }{\partial \qp_i} \quad \text{and} \quad 
  D_{\qH_\rb} = \frac{\partial \qH_\rb}{\partial \qp_i}\frac{\partial}{\partial q^i}
\end{displaymath}
with $D_{\qH} = D_{\qH_\ra} + D_{\qH_\rb}$. Notice that, taken separately, each vector field can be trivially integrated
exactly. For $(q_\alpha,\qp_\alpha) \in T^*\MG$ (we use the subindex $\alpha$ to refer to an element in a sequence, not
a particular coordinate) we have 
\begin{displaymath}
 \begin{pmatrix}
  \dot{q} \\ \dot{\qp} 
 \end{pmatrix}
 =
 \begin{pmatrix}
  0 \\
  -\nabla_q \qH_\ra(q)
 \end{pmatrix}
 \quad \Longrightarrow \quad \begin{pmatrix} q_{\alpha+1} \\ \qp_{\alpha+1} \end{pmatrix} =
  \begin{pmatrix} q_\alpha \\ \qp_\alpha - h\cdot \nabla_q \qH_\ra(q_\alpha) \end{pmatrix} = \phi_{\ra,h}
  \begin{pmatrix}
   q_\alpha \\ \qp_\alpha
  \end{pmatrix}
\end{displaymath}
and
\begin{displaymath}
 \begin{pmatrix}
  \dot{q} \\ \dot{\qp} 
 \end{pmatrix}
 =
 \begin{pmatrix}
  +\nabla_p \qH_\rb(\qp) \\
  0 
 \end{pmatrix}
 \quad \Longrightarrow \quad \begin{pmatrix} q_{\alpha+1} \\ \qp_{\alpha+1} \end{pmatrix} = 
  \begin{pmatrix} q_\alpha + h\cdot \nabla_\qp \qH_\rb(\qp_\alpha) \\ \qp_\alpha \end{pmatrix} = \phi_{\rb,h}
   \begin{pmatrix}
    q_\alpha \\ \qp_\alpha
   \end{pmatrix} \;.
\end{displaymath}
A first-order symplectic method can be easily constructed by performing the composition
\begin{equation} \label{eq:split}
 \psi_h = \phi_{\rb,h} \circ \phi_{\ra,h} \;.
\end{equation}
Indeed, the maps $\phi_{\rb,h}$ and $\phi_{\ra,h}$ are symplectic because they are the exact flows of Hamiltonian vector
fields. Since the composition of two symplectic maps is again symplectic, $\psi_h$ is symplectic.

Many simple Hamiltonian systems with phase space $T^*\MG$ are \emph{not} governed by separable Hamiltonians,
so splitting methods cannot be applied directly. However, the Hamiltonian function of many mechanical systems
becomes separable if the phase space is embedded in $T^*\RE^n$ (see, e.g., Remark~\ref{rem:sepa}). An interesting
symplectic method that is particularly well suited for this class of systems was proposed in~\cite{reich1996}.
Roughly speaking, the idea is to compute a symplectic method for the \emph{unconstrained} Hamiltonian vector field
$D_H : T^*\RE^n \to T(T^*\RE^n)$, i.e., a symplectic map $\psi_{H,h}$ approximating the solution of the \ac{ODE}
(notice the absence of the constraint equations)
\begin{align*}
 \dot{r} &= +\nabla_p H(r,p) \\
 \dot{p} &= -\nabla_r H(r,p)
\end{align*}
at time $t = h$.

If $H$ is separable, $\psi_{H,h}$ can be readily found. The method $\Psi_{H,g,h}$ for the original constrained Hamiltonian vector field
$X_{H,g}$ is then constructed by taking the image of $\psi_{H,h}$ and applying a correction term that ensures that the value of
$\Psi_{H,g,h}$ belongs to $L \MG$, so that the constraints are satisfied. The correction is done in a careful way so that the
resulting map is still symplectic (see also~\cite{leimkuhler}). Depending on the accuracy of $\psi_{H,h}$, the resulting $\Psi_{H,g,h}$ can
be of first or second order (see \S~\ref{sec:main} for details).

\subsection{Symmetric methods}

For each $t$ for which the solution is defined, the flow $\phi_t(x_0)$ of an autonomous differential equation defines
a transformation on the phase space. It follows from the group property of the flow~\cite{arnold4} that the inverse of the transformation can
be obtained simply by reversing time, that is, $\phi_t^{-1}(x_1) = \phi_{-t}(x_1) = x_0$. Needless to say, this property does not hold
in general for a discrete approximation $\psi_h$, which motivates the  following definition.

\begin{definition}
The \emph{adjoint method} $\psi_h^*$ of a method $\psi_h$ is the inverse map of the original method with reversed time step $-h$, i.e.,
\begin{displaymath}
 \psi_h^* := (\psi_{-h})^{-1} \;.
\end{displaymath}
In other words, $x_1 = \psi_h^*(x_0)$ is implicitly defined by $\psi_{-h}(x_1) = x_0$. A method for which $\psi_h^* = \psi_h$ is called
\emph{symmetric}.
\end{definition}

From a theoretical point of view, an approximate discrete-time flow should be symmetric because actual continuous flows are. But symmetry is
important from a practical point of view too. It has been proved in~\cite{yoshida1990} that all symmetric methods are of even order, a fact
that can be exploited to construct high-order methods from simple lower-order methods. For example, one can take a first-order
non-symmetric method, compute its adjoint and construct a symmetric method
\begin{equation} \label{eq:sym}
 \Psi_h = \psi_{\frac{h}{2}}\circ\psi_{\frac{h}{2}}^* \;.
\end{equation}
We know that $\Psi_h$ is at least first order, but since $\Psi_h$ is symmetric, we also know that the order has to be even, so we conclude
that the method is actually of second order.

The scheme~\eqref{eq:sym} works particularly well for splitting methods. Take, e.g., the integration
scheme~\eqref{eq:split}. The maps $\phi_{\rb,h}$ and $\phi_{\ra,h}$ are symmetric (because they are exact solutions of a
differential equation), but their composition is not symmetric in general. To remedy this, one can compute the adjoint method
\begin{equation} \label{eq:adj}
 \psi^*_h = \left(\phi_{\rb,-h}\circ \phi_{\ra,-h}\right)^{-1} = 
  \phi_{\ra,-h}^{-1} \circ \phi_{\rb,-h}^{-1} = 
   \phi_{\ra,h} \circ \phi_{\rb,h}
\end{equation}
and, using~\eqref{eq:split},~\eqref{eq:adj} and~\eqref{eq:sym}, construct
\begin{displaymath}
 \Psi_h = \phi_{\rb,\frac{h}{2}}\circ \phi_{\ra,\frac{h}{2}} \circ \phi_{\ra,\frac{h}{2}} \circ
  \phi_{\rb,\frac{h}{2}} = \phi_{\rb,\frac{h}{2}}\circ \phi_{\ra,h} \circ \phi_{\rb,\frac{h}{2}} \;,
\end{displaymath}
which is a second-order symmetric method.

\subsection{Modified vector fields and exponential representations}

Consider a vector-valued function $F$ and a vector field $X$, both defined on $L \MG$. If $F$ and $X$ are analytic, then the composition of $F$ and
the generated flow $\phi_t(x_0)$ can be expanded in a Taylor series around $t = 0$,
\begin{displaymath}
 F\circ\phi_t(x_0) = \exp(tX)F(x_0) := \sum_{i=0}^\infty \frac{t^i}{i!} X^i(F)(x_0) \;,
\end{displaymath}
where $X^0(F) = F$, $X^2(F) = X(X(F))$, $X^3(F) = X(X^2(F))$ etc (see~\cite{olver,varadarajan} for details). In particular, if $F$ is taken
as the identity function $\id$, one obtains the flow
$\phi_t(x_0) = \exp(tX)\id(x_0)$.
Since an $s$-order method $\psi_h$ for $X$ coincides with the flow of a modified vector field
$\tilde{X} = X + \cO(h^s)$~\cite[p. 340]{hairer}, it is also possible to expand $\psi_h$ in a Taylor series,
$\psi_h(x_0) = \exp(h\tilde{X})\id(x_0)$.
This exponential notation is a convenient way to express the relationship between a vector field and the flow generated by it, as well as 
to analyze the composition of flows.

\section{A splitting method for implicit \acl{PH} systems} \label{sec:main}

Suppose that there is a sequence of commands $\{u\indices{_l_\alpha}\}_{\alpha \in \NA}$. Each command in the sequence arrives at the
discrete instants of time $\alpha = 0,h,2h,\dots$ where $h$ is a positive real number ---such commands could be generated, e.g., by a computer
program. Suppose further that a zero-order hold transforms this sequence into piece-wise constant controls
\begin{equation} \label{eq:zoh}
 u_l(t) \equiv u\indices{_l_\alpha} \quad \text{for } t \in [\alpha h,\alpha h+h) \;,
\end{equation}
which are fed into a \ac{PH} system. Let $\phi_t(x_0,u(\cdot))$ be the integral curve of the non-autonomous vector field~\eqref{eq:XHug} with
control~\eqref{eq:zoh} and passing through $x_0 \in L \MG$ at $t = 0$. Let
\begin{displaymath}
 y^l(t) = \left(U\indices{_i^l} \frac{\partial H}{\partial p_i}\right)\circ\phi_t(x_0,u(\cdot))
\end{displaymath}
be the corresponding outputs and let $\{y\indices{^l_\alpha}\}$ with $y\indices{^l_\alpha} := y^l(\alpha h)$ be the sequence obtained by
sampling them at discrete instants of time $\alpha h$ (see Fig.~\ref{fig:ZOH}). We call the resulting system a \emph{sampled-data \acl{PH}
system}.

\begin{figure}
\begin{center}
 \includegraphics[width=0.85\textwidth]{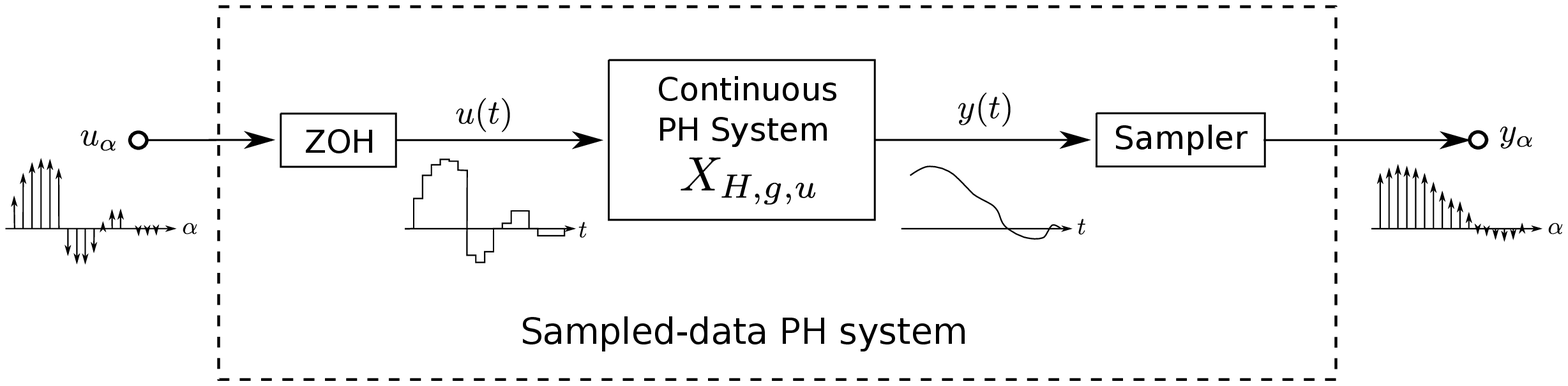}
\end{center}
\caption{Sampled-data \ac{PH} system with sampling period $h$. The zero-order hold produces piece-wise constant inputs
$u_l(t) \equiv u\indices{_l_\alpha}$ for $t \in [\alpha h,\alpha h+h)$, which is fed to the continuous-time \ac{PH} system. The output is then
sampled to generate the discrete-time output sequences $\{y\indices{^l_\alpha}\}$.}
\label{fig:ZOH}
\end{figure}

The goal here is to develop a method for derivation of discrete-time (or sampled-data) models for \ac{PH} systems given by implicit
vector fields. The underlying idea is to split the \ac{PH} vector field into two components: the vector field describing an unconstrained
system with state-space equal to the whole $T^*\RE^n$, and a vector field containing the Lagrange multipliers, the one that maintains
the trajectories on the submanifold $L \MG$. Splitting the vector field simplifies the computation of the sampled-data model by
decomposing the problem into two simpler subproblems.

Now we extend the results of~\cite{reich1996} to the~\ac{PH} case. We show that, with a 
straightforward modification, the method presented in~\cite{reich1996}, originally intended as an
integration scheme for autonomous Hamiltonian systems, can be used to compute sampled-data models that
preserve the main properties of a \ac{PH} system.

Consider again the implicit vector field
\begin{equation} \label{eq:XHug2}
 X_{H,u,g} = D_H + \left(u_lU\indices{_i^l} - \lambda_j \frac{\partial g^j}{\partial r^i} \right)\frac{\partial }{\partial p_i} \;,
  \quad g = 0 \;, \quad y^l = U\indices{_i^l}\frac{\partial H}{\partial p_i} \;,
\end{equation}
defined on $L \MG$, with $D_H$ as in~\eqref{eq:DH} and with piece-wise constant controls~\eqref{eq:zoh}. Suppose that a method
$\psi_{H,u,h} : T^*\RE^n \to T^*\RE^n$ of order $s \ge 1$ for the unconstrained vector field
$X_{H,u} = D_H + u_lU\indices{_i^l} \frac{\partial }{\partial p_i}$
has been computed. Again, in many cases $H$ is separable so a high-order and symmetric method with a \ac{PH} modified vector field can be
easily found. The controls $u_l$ are constant during each sampling interval, which further simplifies the task of finding $\psi_{H,u,h}$.\\

Let us define the map
\begin{equation} \label{eq:proj}
 \Phi_{\Lambda,h} \begin{pmatrix} r_\alpha \\ p_\alpha \end{pmatrix} := 
  \begin{pmatrix}
   r_\alpha \\ p_\alpha - h G(r_\alpha)^\top \Lambda
  \end{pmatrix} \;.
\end{equation}
Loosely, this is an approximation of the integral curve of the remnant vector field 
$-\lambda_j \frac{\partial g^j}{\partial r^i}\frac{\partial }{\partial p_i}$
evaluated at $t=h$ and subject to $g = 0$. More precisely, for arbitrary functions $\lambda_j$ of $r$ and $p$, we have that
\begin{displaymath}
 D_{\lambda_j g^j} := \frac{\partial (\lambda_j g^j)}{\partial p_i}\frac{\partial }{\partial r^i} -
  \frac{\partial (\lambda_j g^j)}{\partial r_i}\frac{\partial }{\partial p^i} =
 g^j\frac{\partial \lambda_j}{\partial p_i} \frac{\partial }{\partial r^i} - 
   \left( \lambda_j\frac{\partial g^j}{\partial r_i} + g^j\frac{\partial \lambda_j}{\partial r_i} \right) \frac{\partial }{\partial p^i}
\end{displaymath}
and, for $r \in \MG$, the vector field reduces to
\begin{equation} \label{eq:Dlg}
 D_{\lambda_j g^j} = \lambda_j D_{g^j} = -\lambda_j \frac{\partial g^j}{\partial r^i}\frac{\partial }{\partial p_i} \;.
\end{equation}
In other words, when restricted to $\MG$, the vector field $-\lambda_j \frac{\partial g^j}{\partial r^i}\frac{\partial }{\partial p_i}$
is Hamiltonian (hence it generates a symplectic flow).

\begin{lemma}\cite{reich1996} \label{lem:remn}
 Let $g(r_\alpha)=0$. Then, the map~\eqref{eq:proj} is a first-order symplectic method for $D_{\Lambda_j g^j}$. That is, for $r_\alpha \in \MG$,
 \begin{displaymath}
  \Phi_{\Lambda,h}
   \begin{pmatrix}
    r_\alpha \\ p_\alpha
   \end{pmatrix}
   = \exp(h\tilde{D}_{\Lambda_j g^j})\id
   \begin{pmatrix}
    r_\alpha \\ p_\alpha
   \end{pmatrix}
   \;, \quad \tilde{D}_{\Lambda_j g^j} = D_{\tilde{\Lambda}_j g^j}
 \end{displaymath}
 with $\tilde{\Lambda}$ a modified or perturbed version of $\Lambda$.
\end{lemma}

A method for~\eqref{eq:XHug2} can be obtained from the symmetric composition
\begin{equation} \label{eq:reich}
 \Psi_{H,u,g,h} = \Phi_{\mu,\frac{h}{2}} \circ \psi_{H,u,h} \circ \Phi_{\nu,\frac{h}{2}} \;.
\end{equation}
For each $(r_\alpha,p_\alpha) \in L \MG$, the values of $\mu$ and $\nu$ are determined \emph{implicitly} by the constraints
$g(r_{\alpha + 1}) = 0$ and $f(r_{\alpha+1},p_{\alpha+1}) = 0$ (i.e., by $(r_{\alpha+1},p_{\alpha + 1}) \in L \MG$).
In this way, $\Psi_{H,u,g,h}$ defines a transformation on $L \MG$.

The transformation $\Psi_{H,u,g,h}$ produces an approximate discrete flow for a given command sequence $\{u\indices{_l_\alpha}\}$. From
this flow, an approximate output sequence can be obtained by evaluating the output function
$y^l = U\indices{_i^l} \frac{\partial H}{\partial p_i}$ at each discrete time $\alpha h$.

\begin{theorem} \label{thm:main}
 Consider the implicit method $\Psi_{H,u,g,h}$~\eqref{eq:reich} and let $\tilde{X}_{H,u}$ be the modified vector field of $\psi_{H,u,h}$. 
 \begin{romannum}
  \item \label{it:order} The method preserves the constraints $g^j = 0$, $f^j = 0$ and is of order $\bar{s} = \min(s,2)$, where $s$ is the order
   of $\psi_{H,u,h}$.
  \item \label{it:syme} The method is symmetric if $\psi_{H,u,h}$ is symmetric.
  \item \label{it:symp} If $\psi_{H,u,h}$ is symplectic for $u\indices{_l_\alpha} \equiv 0$ (i.e., if $\tilde{X}_{H,u}$ is \acl{PH}), then
   the modified vector field $\tilde{X}_{H,u,g} : L \MG \to T(L \MG)$ is also \acl{PH} with Hamiltonian and output functions
   \begin{equation} \label{eq:HO}
    \tilde{H} = H + \cO(h^{\bar{s}}) \quad \text{and} \quad \tilde{y}^l = y^l + \cO(h^{\bar{s}}) \;.
   \end{equation}
 \end{romannum}
\end{theorem}

\begin{proof}
 The method preserves the constraints by construction. The proof about the order of the method follows the
 same lines as the one given in~\cite{reich1996} except that, since we are dealing with \acl{PH} vector fields,
 Lie brackets have to be used instead of Poisson brackets. We will compute $\tilde{X}_{H,u,g}$, the modified
 vector field generating $\Psi_{H,u,g,h}$, and show that it agrees with $X_{H,u,g}$ up to the first or second order,
 depending on whether $\Psi_{H,u,g,h}$ is, respectively, first or second order.

 Let us consider the case $s = 1$. Using the exponential notation and Lemma~\ref{lem:remn}, the composition~\eqref{eq:reich} takes the form
 \begin{displaymath}
  \Psi_{H,u,g,h} = \exp\left( \frac{h}{2}D_{\tn_j g^j} \right)\exp\left( h \tilde{X}_{H,u} \right)
   \exp\left( \frac{h}{2}D_{\tm_j g^j} \right)\id \;,
 \end{displaymath}
 where $\tilde{X}_{H,u} = X_{H,u} + \cO(h)$. Applying the \ac{BCH} formula~\cite{varadarajan} to the product of the first two factors and
 truncating after the first term gives 
 \begin{equation} \label{eq:prod1}
  \exp\left( \frac{h}{2}D_{\tn_j g^j} \right)\exp\left( h \tilde{X}_{H,u} \right) = \exp\left( h \tilde{X}' \right)
 \end{equation}
 with $\tilde{X}' = \tilde{X}_{H,u} + \frac{1}{2}D_{\tn_j g^j} + \cO(h)$. Applying \ac{BCH} again to include the third factor gives
 \begin{equation} \label{eq:prod2}
  \exp\left( h \tilde{X}' \right)\exp\left( \frac{h}{2}D_{\tm_j g^j} \right) = \exp\left( h \tilde{X}_{H,u,g} \right)
 \end{equation}
 with the modified vector field $\tilde{X}_{H,u,g} = \tilde{X}_{H,u} + \frac{1}{2}\left( D_{\tn_j g^j} + D_{\tm_j g^j} \right) + \cO(h)$.
 Using~\eqref{eq:Dlg} and $\tilde{X}_{H,u} = X_{H,u} + \cO(h)$, we can write the modified vector field as
 \begin{equation} \label{eq:Xmod1}
  \tilde{X}_{H,u,g} = X_{H,u} + \frac{\tn_j + \tm_j}{2} D_{g^j} + \cO(h) \;.
 \end{equation}

 The hidden constraints $f^l = 0$ imply that
 \begin{equation} \label{eq:fut}
  \tilde{X}_{H,u,g}(f^l) = X_{H,u}(f^l) + \frac{\tn_j + \tm_j}{2} D_{g^j}(f^l) + \cO(h) = 0 \;.
 \end{equation}
 It follows from~\eqref{eq:fut},~\eqref{eq:Dlg} and~\eqref{eq:fu}, that the Lagrange multipliers $\lambda_j$ and the `modified Lagrange
 multipliers' $\tilde{\nu}_j$ and $\tilde{\mu}_j$ are related by the equation
 $(\tn_j + \tm_j)/2 = \lambda_j + \cO(h)$,
 which when substituted back in~\eqref{eq:Xmod1} gives the desired result:
 \begin{displaymath}
  \tilde{X}_{H,u,g} = X_{H,u} + \lambda_j D_{g^j} + \cO(h) = X_{H,u,g} + \cO(h) \;.
 \end{displaymath}

 For $s = 2$ we follow the same procedure, but we truncate the \ac{BCH} formula after the second term. For the expression~\eqref{eq:prod1},
 the intermediate vector field is
 \begin{displaymath}
 \tilde{X}' = \tilde{X}_{H,u} + \frac{1}{2}D_{\tn_j g^j} + \frac{h}{4}\left[ D_{\tn_j g^j} , \tilde{X}_{H,u} \right] + \cO(h^2)
 \end{displaymath}
 where $[\cdot,\cdot]$ is the standard Lie bracket. Using the initial assumption $\tilde{X}_{H,u} = X_{H,u} + \cO(h^2)$,  we can write
 $\tilde{X}'$ as
 \begin{displaymath}
  \tilde{X}' = X_{H,u} + \frac{1}{2}D_{\tn_j g^j} + \frac{h}{4}\left[ D_{\tn_j g^j} , X_{H,u} \right] + \cO(h^2) \;.
 \end{displaymath}
 Regarding~\eqref{eq:prod2}, the modified vector field for the complete scheme is
 \begin{multline*}
  \tilde{X}_{H,u,g} = X_{H,u} + \frac{1}{2}\left( D_{\tn_j g^j} + D_{\tm_j g^j} \right) \\
   + \frac{h}{4}\left[ D_{\tn_j g^j} , X_{H,u} \right] + 
      \frac{h}{4}\left[ X_{H,u} + \frac{1}{2}D_{\tn_j g^j} , D_{\tm_j g^j} \right] + \cO(h^2) \;.
 \end{multline*}
 Using~\eqref{eq:Dlg} and the skew symmetry and bilinearity of the Lie bracket, the vector field can be equivalently written as
 \begin{multline} \label{eq:prodt}
  \tilde{X}_{H,u,g} = X_{H,u} + \frac{\tn_j + \tm_j}{2} D_{g^j} + \frac{h}{4}\left[(\tn_j - \tm_j) D_{ g^j} , X_{H,u} \right] \\
   + \frac{h}{8}\left[ \tn_j D_{g^j} , \tm_j D_{g^j} \right] + \cO(h^2) \;.
 \end{multline}

 In order to extract information from the equation $\tilde{X}_{H,u,g}(g^l) = 0$, let us first open the brackets in~\eqref{eq:prodt} and
 write
 \begin{multline*}
  \tilde{X}_{H,u,g} = X_{H,u} + \frac{\tn_j + \tm_j}{2} D_{g^j} + \frac{h}{4}(\tn_j - \tm_j) D_{ g^j} X_{H,u} \\
   - \frac{h}{4}X_{H,u}\left((\tn_j - \tm_j) D_{ g^j} \right) 
   + \frac{h}{8}\tn_j D_{g^j}\left( \tm_j D_{g^j} \right) - \frac{h}{8}\tm_j D_{g^j}\left( \tn_j D_{g^j} \right) + \cO(h^2) \;.
 \end{multline*}
 Taking into account that $f^l = X_{H,u}(g^l) = 0$ and $D_{g^j}(g^l) \equiv 0$, we have that
 \begin{displaymath}
  \tilde{X}_{H,u,g}(g^l) = \frac{h}{4}(\tn_j - \tm_j) D_{ g^j} f^l + \cO(h^2) = 0 \;,
 \end{displaymath}
 from which we can see that modified Lagrange multipliers satisfy the order relation
 \begin{equation} \label{eq:nu_mu}
  \tn_j - \tm_j = \cO(h) \;.
 \end{equation}
 By substituting~\eqref{eq:nu_mu} back in~\eqref{eq:prodt} we can verify that the commutators are actually second order, that is,
 \begin{align}
  \tilde{X}_{H,u,g} & = X_{H,u} + \frac{\tn_j + \tm_j}{2} D_{g^j} +
   \frac{h}{8}\left[ \tn_j D_{g^j} , \left(\tn_j + \cO(h)\right) D_{g^j} \right] + \cO(h^2) \nonumber \\
                    & = X_{H,u} + \frac{\tn_j + \tm_j}{2} D_{g^j} + \cO(h^2) \;. \label{eq:Xmod2}
 \end{align}
 From $\tilde{X}_{H,u,g}(f^l) = 0$ and~\eqref{eq:fu} we conclude that, when $s = 2$,
 $(\tn_j + \tm_j)/2 = \lambda_j + \cO(h^2)$,
 so the desired result follows:
 $\tilde{X}_{H,u,g} = X_{H,u} + \lambda_j D_{g^j} + \cO(h^2) = X_{H,u,g} + \cO(h^2)$.

 For statement~(ii), notice that, when restricted to $L \MG$, the method~\eqref{eq:reich} can be
 described by the implicit equations
 \begin{subequations} \label{eq:impl}
 \begin{align}
  \begin{pmatrix}
   r_{\alpha+1} \\ p_{\alpha+1} + h G(r_{\alpha+1})^\top \mu
  \end{pmatrix}
   & = \psi_{H,u,h}
   \begin{pmatrix}
    r_\alpha \\ p_\alpha - h G(r_{\alpha})^\top \nu
   \end{pmatrix} \label{eq:psi} \\
                 g(r_{\alpha + 1}) & = g(r_{\alpha}) \\
  f(r_{\alpha + 1},p_{\alpha + 1}) & = f(r_{\alpha},p_{\alpha}) \;,
 \end{align}
 \end{subequations}
 where $r_\alpha, p_\alpha$ are the independent variables and $r_{\alpha+1}, p_{\alpha + 1}$ are the dependent variables. The vectors
 $\nu, \mu$ are (also dependent) dummy variables that can be discarded after $r_{\alpha+1}, p_{\alpha + 1}$ have been found.

 After reversing time (that is, after substituting $h$ by $-h$), equation~\eqref{eq:psi} becomes
 \begin{displaymath}
  \begin{pmatrix}
   r_{\alpha+1} \\ p_{\alpha+1} - h G(r_{\alpha+1})^\top \mu
  \end{pmatrix}
   = \psi_{H,u,-h}
   \begin{pmatrix}
    r_\alpha \\ p_\alpha + h G(r_{\alpha})^\top \nu
   \end{pmatrix} \;.
 \end{displaymath}
 Recall that $\psi_{H,u,-h} =\psi_{H,u,h}^{-1}$ if $\psi_{H,u,h}$ is symmetric. Therefore, when restricted to $L \MG$, the reverse-time
 method is
 \begin{subequations} \label{eq:impl_h}
 \begin{align}
  \begin{pmatrix}
   r_\alpha \\ p_\alpha + h G(r_{\alpha})^\top \nu
  \end{pmatrix}
   & = \psi_{H,u,h}
  \begin{pmatrix}
   r_{\alpha+1} \\ p_{\alpha+1} - h G(r_{\alpha}+1)^\top \mu
  \end{pmatrix} \\
                 g(r_{\alpha + 1}) & = g(r_{\alpha}) \\
  f(r_{\alpha + 1},p_{\alpha + 1}) & = f(r_{\alpha},p_{\alpha}) \;,
 \end{align}
 \end{subequations}
 which is the same as~\eqref{eq:impl}, but with $r_\alpha,p_\alpha$ and $\nu$ interchanged with $r_{\alpha+1},p_{\alpha+1}$ and $\mu$,
 respectively. This implies that, if we input $r_{\alpha+1}, p_{\alpha+1}$ as independent variables, we recover $r_{\alpha}, p_{\alpha}$
 as the dependent variables, that is: $\Psi_{H,u,g,-h}$ is the inverse mapping of $\Psi_{H,u,g,h}$. (In general, the vectors $\nu, \mu$
 obtained using~\eqref{eq:impl} will be different from those obtained using~\eqref{eq:impl_h}, but this is inconsequential since these are
 dummy variables.)

 In statement~(iii), the fact that $\tilde{X}_{H,u,g}$ is \ac{PH} follows directly from Lemma~\ref{lem:remn}, Definition~\ref{def:PH}
 and the fact that the composition of symplectic maps is again symplectic. In other words, $\Psi_{H,u,g,h}$ is symplectic when
 $u\indices{_l_\alpha} \equiv 0$, so
 \begin{equation} \label{eq:nona}
  \tilde{X}_{H,u,g} = X_0 + u_lZ^l = D_{\tilde{H}} + u_lZ^l + \tilde{\lambda}^jD_{g^j}
 \end{equation}
 for some Hamiltonian function $\tilde{H}$ and some input vector fields $Z^l$. Since the method is of order $\bar{s}$, we have
 \begin{equation} \label{eq:sb}
  \tilde{X}_{H,u,g} = D_H + u_lU\indices{_i^l}\frac{\partial H}{\partial p_i} + \lambda^jD_{g^j} + \cO(h^{\bar{s}}) \;.
 \end{equation}
 By setting $u_l = 0$ and recalling that $\tilde{\lambda}_j = \lambda_j + \cO(h^{\bar{s}})$, it follows from~\eqref{eq:nona}
 and~\eqref{eq:sb} that $D_{\tilde{H}} = D_H + \cO(h^{\bar{s}})$, which implies that $\tilde{H} = H + \cO(h^{\bar{s}})$, and, in turn, that
 $Z^l = U\indices{_i^l}\frac{\partial}{\partial p_i} + \cO(h^{\bar{s}})$.
 The modified output functions are thus
 $\tilde{y}^l = Z^l(\tilde{H})$, that is, $\tilde{y}^l = U\indices{_i^l}\frac{\partial }{\partial p_i}\big(H + \cO(h^{\bar{s}})\big) + \cO(h^{\bar{s}}) = 
                 y^l + \cO(h^{\bar{s}})$.
\end{proof}

Let us now turn to the problem of energy balance under sample and hold. The power balance~\eqref{eq:pb} implies that
\begin{equation} \label{eq:spb}
 H_{\alpha + 1} - H_\alpha = \int_{\alpha h}^{\alpha h + h} u_l(t)y^l(t) \rd t \;,
\end{equation}
where we have defined the sampled Hamiltonian $H_\alpha := H(x_\alpha)$. A usual way to improve the transient behavior of the system is to add
damping by means of a continuous control law ~\cite{ortega2001}
\begin{equation} \label{eq:damp}
 u_l(t) = -K\indices{_l_j}y^j(t) \;,
\end{equation}
with $\{ K\indices{_l_j} \}$ a symmetric and positive semi-definite matrix. With the control law
~\eqref{eq:damp}, the power
balance~\eqref{eq:spb} results in the dissipation inequality
$H_{\alpha + 1} - H_\alpha \le 0$,
which guarantees that $H_\alpha$ decreases monotonically and, if the right conditions are met, the system converges to a state of minimal energy.\\

Suppose that the output is being sampled and that the input is being held at intervals of length $h$. The control sequence is then given by
\begin{equation} \label{eq:dampd}
 u\indices{_l_\alpha} = -K\indices{_l_j}y\indices{^j_\alpha}
\end{equation}
and the power balance~\eqref{eq:spb} takes the form
\begin{displaymath}
 H_{\alpha + 1} - H_\alpha = u\indices{_l_\alpha} \int_{\alpha h}^{\alpha h + h} y^l(t) \rd t =
  u\indices{_l_\alpha} \int_{0}^{h} y^l(\alpha h + \tau) \rd \tau \;.
\end{displaymath}
Applying Taylor's theorem to the integral term gives
\begin{displaymath}
 H_{\alpha + 1} - H_\alpha = \sum_l u\indices{_l_\alpha} \left( y\indices{^l_\alpha} h + \cO(h^2) \right) =
  -h K\indices{_l_j}y\indices{^l_\alpha}y\indices{^j_\alpha} + \sum_l u\indices{_l_\alpha} \cO(h^2) \;,
\end{displaymath}
so $H_\alpha$ decreases when $h$ is small enough and the norm of $y_\alpha$ is large enough.

Since the approximate sampled-data model~\eqref{eq:reich} is also \ac{PH} (cf. item~(iii) of Theorem~\ref{thm:main}), it satisfies
(again, after applying Taylor's theorem)
\begin{equation} \label{eq:HOa}
 \tilde{H}_{\alpha + 1} - \tilde{H}_\alpha = \sum_l u\indices{_l_\alpha} \left( \tilde{y}\indices{^l_\alpha} h + \cO(h^2) \right)
\end{equation}
for some $\tilde{H}$ and $\tilde{y}$. According to~\eqref{eq:HO}, the energy balance~\eqref{eq:HOa} takes the form
\begin{displaymath}
 H_{\alpha+1} - H_\alpha + \cO(h^{\bar{s}}) = \sum_l u\indices{_l_\alpha} \left( \left( y\indices{^l_\alpha} + \cO(h^{\bar{s}})\right) h 
 + \cO(h^2) \right) \;.
\end{displaymath}

The same control sequence~\eqref{eq:dampd} produces
\begin{displaymath}
 H_{\alpha+1} - H_\alpha = -h K\indices{_l_j}y\indices{^l_\alpha}y\indices{^j_\alpha} + \sum_l u\indices{_l_\alpha} \cO(h^2) + \cO(h^{\bar{s}}) \;.
\end{displaymath}
Thus, for $\bar{s} = 2$, the qualitative behavior of the approximated sampled data model is the same as the exact one: $H_\alpha$ decreases when
$h$ is small enough and the norm of $y_\alpha$ is large enough.

\subsection*{Example: A double planar pendulum (continued)}

Let us compute a sampled-data model for the double pendulum described in the previous examples. The first step is to
compute a sample-data model for the simple unconstrained \ac{PH} system
$X_{H,u} = D_H + u_lU\indices{_i^l} \frac{\partial }{\partial p_i}$,
where $H$ is given by~\eqref{eq:HEx} and $U\indices{_i^l}$ by~\eqref{eq:U}. The unconstrained and unactuated Hamiltonian vector
field $D_H$ describes a pair of masses with initial positions $r\indices{^a_0}$ and $r\indices{^b_0}$ and initial momenta $p\indices{_a_0}$
$p\indices{_b_0}$, simply falling under the influence of gravity. The exact flow generated by $D_H$, the drift, denoted by
$(r_{\alpha+1}, p_{\alpha+1}) = \phi_{H,h}(r_{\alpha}, p_{\alpha})$, is then given by
\begin{align*}
 r\indices{^{a_x}_{\alpha+1}} &= r\indices{^{a_x}_\alpha} + \frac{h}{m_a}p\indices{_{a_x}_\alpha} \;, \quad 
 r\indices{^{a_y}_{\alpha+1}} = r\indices{^{a_y}_{\alpha}} + \frac{h}{m_a}p\indices{_{a_y}_\alpha} - \bar{g}\frac{h^2}{2} \\
 r\indices{^{b_x}_{\alpha+1}} &= r\indices{^{b_x}_\alpha} + \frac{h}{m_b}p\indices{_{b_x}_\alpha} \;, \quad
 r\indices{^{b_y}_{\alpha+1}} = r\indices{^{b_y}_{\alpha}} + \frac{h}{m_b}p\indices{_{b_y}_\alpha} - \bar{g}\frac{h^2}{2}
\end{align*}
and
\begin{align*}
 p\indices{_{a_x}_{\alpha+1}} &= p\indices{_{a_x}_\alpha} \;, \quad 
 p\indices{_{a_y}_{\alpha+1}} = p\indices{_{a_y}_{\alpha}} - m_a\bar{g}h \\
 p\indices{_{b_x}_{\alpha+1}} &= p\indices{_{b_x}_\alpha} \;, \quad
 p\indices{_{b_y}_{\alpha+1}} = p\indices{_{b_y}_{\alpha}} - m_b\bar{g}h \;.
\end{align*}
The exact flow generated by $u_lU\indices{_i^l} \frac{\partial }{\partial p_i}$, the control vector field without drift, is denoted by
$(r_{\alpha+1}, p_{\alpha+1}) = \phi_{u,h}(r_{\alpha}, p_{\alpha})$. It is given by
\begin{displaymath}
 r\indices{^i_{\alpha+1}} = r\indices{^i_{\alpha}} \;, \quad i \in \left\{a_x,a_y,b_x,b_y \right\}
\end{displaymath}
and
\begin{align*}
 p\indices{_{a_x}_{\alpha+1}} &= p\indices{_{a_x}_\alpha} - \frac{h}{l_a^2}(u_1 - u_2)r\indices{^{a_y}_\alpha} \;, \quad
 p\indices{_{a_y}_{\alpha+1}} = p\indices{_{a_y}_{\alpha}} + \frac{h}{l_a^2}(u_1 - u_2)r\indices{^{a_x}_\alpha} \\
 p\indices{_{b_x}_{\alpha+1}} &= p\indices{_{b_x}_\alpha} - \frac{h}{l_b^2}u_2 r\indices{^{\delta_y}_\alpha} \;, \quad
 p\indices{_{b_y}_{\alpha+1}} = p\indices{_{b_y}_{\alpha}} + \frac{h}{l_b^2}u_2 r\indices{^{\delta_x}_\alpha} \;.
\end{align*}
From \S~\ref{sec:split}, we know that a simple symmetric method of order two for $X_{H,u}$ is
\begin{equation} \label{eq:psi_exa}
 \psi_{H,u,h} = \phi_{H,\frac{h}{2}}\circ\phi_{u,h}\circ\phi_{H,\frac{h}{2}} \;.
\end{equation}

Notice that $\phi_{u,h} = \id$ when $u \equiv 0$, so $\psi_{H,u,h} = \phi_{H,\frac{h}{2}}\circ\phi_{H,\frac{h}{2}} = \phi_{H,h}$, which is a
symplectic map because it is the exact solution of a Hamiltonian system. Therefore, $\tilde{X}_{H,u}$ is \ac{PH}, and,
from Theorem~\ref{thm:main}, it follows that the implicit method~\eqref{eq:reich}, with $\Psi_{H,u,h}$ as
in~\eqref{eq:psi_exa} is the exact solution of a \ac{PH} system $\tilde{X}_{H,u,g}$ with Hamiltonian function $\tilde{H} = H + \cO(h^2)$ and
output function $\tilde{y} = y + \cO(h^2)$ (i.e., $\bar{s} = s = 2$).\\

\begin{table}
\footnotesize
\centering
\begin{tabular}{ll}
 Value         & Description \\
 \hline
 $l_a = 0.6 [\metre]$    & Length of the first link \\
 $l_b = 0.3 [\metre]$    & Length of the second link \\
 $m_a = 0.2 [\kilogram]$ & Value of the first mass \\
 $m_b = 0.6 [\kilogram]$ & Value of the second mass \\
 $\bar{g} = 9.81 [\metre\per\second\square]$ & Acceleration due to gravity
\end{tabular}
\caption{Parameters for the double pendulum.}
\label{tab:param}
\end{table}

\begin{figure}
\begin{center}
\includegraphics[width=0.7\textwidth]{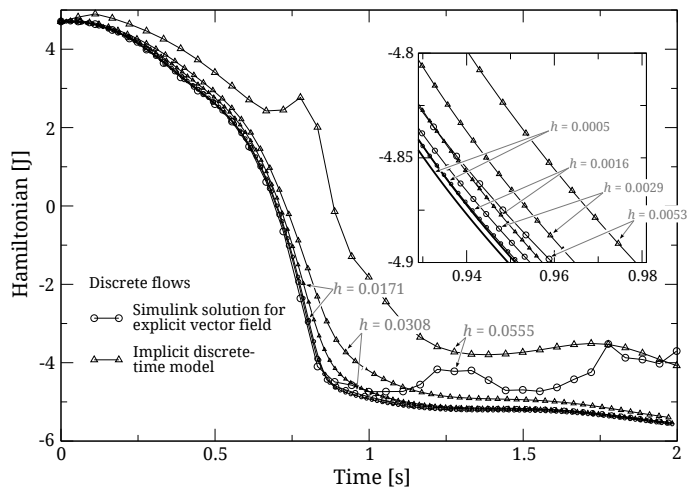}
\end{center}
\caption{Results of the numerical experiment. The Hamiltonian function is plotted against time. The explicit model (simulated using Matlab's
 module Simulink) is compared with the implicit model (simulated using a Matlab script). As expected, $H$ is monotonically decreasing when $h$ is
 small enough and the time series converges as $h$ goes to zero.}
\label{fig:Hvsh}
\end{figure}

The sampled-data model was tested using the parameters shown in Table~\ref{tab:param}. For illustration purposes, we chose a damping control
$u_\alpha = -0.3 \cdot y_\alpha$ and simulated the closed-loop system using the sampled-data model~\eqref{eq:reich}. Figure~\ref{fig:Hvsh}
shows the discrete-time series of $H$ for different values of $h$. It can be seen that the time series converge and, as expected, the value
of $H_\alpha$ decreases monotonically when $h$ is small enough (in this case, less or equal to 30 [\milli\second]). For comparison purposes,
we have included the evolution of $H$ that is obtained by simulating (with Matlab's Simulink) the explicit model developed in~\cite{castanos2013}
in series with a sampler and a zero-order hold.

\section{Conclusions}

We have extended the second-order integration method presented in~\cite{reich1996}. The original method applies to autonomous Hamiltonian systems and,
being symplectic, preserves the Hamiltonian structure of the continuous-time system. The extended method can be applied to port-Hamiltonian
systems, which are Hamiltonian systems equipped with input--output pairs. The extended method preserves the port-Hamiltonian structure.
Affinity in the controls is lost by the method but, fortunately, the passivity properties of the continuous-time system can be recovered by 
a suitable redefinition of the output. Interestingly, the relation between the original and the new output is also of second order.

The integration method can be used with the purposes of numerical simulation or with the purpose of deriving discrete-time models to be used
in the design of discrete-time control laws.

\bibliography{integrators,extras}

\end{document}